\documentclass{article}

\usepackage{arxiv}

\usepackage[utf8]{inputenc} 
\usepackage[T1]{fontenc}    
\usepackage{hyperref}       
\usepackage{url}            
\usepackage{booktabs}       
\usepackage{amsfonts}       
\usepackage{nicefrac}       
\usepackage{microtype}      
\usepackage{lipsum}
\usepackage{fancyhdr}
\usepackage{mathrsfs}

\usepackage{graphicx}          
\usepackage{amsmath} 
\usepackage{amssymb}  

\usepackage{amsthm}
\usepackage{color,soul}
\usepackage{cite}
\usepackage{tabularx}
\usepackage{subcaption}
\usepackage{enumitem}
\usepackage{wrapfig}

\newtheorem{thmm}{Theorem}
\newtheorem{propp}{Proposition}
\newtheorem{remm}{Remark}

\newtheorem{Lemm}{Lemma}

\newcommand{\hu}{|h|_{\mathscr{U}}}
\newcommand{\bdot}{\dot{\mathscr{B}}(h)}
\newcommand{\bh}{{\mathscr{B}}(h)}

\newcommand*\diff{\mathop{}\!\mathrm{d}} 

\title{An Input-to-State Safety Approach Towards Safe Control of a Class of Parabolic PDEs Under Disturbances}

\author{Tanushree~Roy,
Ashley~Knichel, 
Satadru~Dey, 
\thanks{Tanushree Roy is with the Department
of Mechanical Engineering, Texas Tech University, Lubbock, Texas, 79409,  e-mail: (tanushree.roy@ttu.edu). USA. Ashley Knichel and Satadru Dey are with the Department
of Mechanical Engineering, The Pennsylvania State University, University Park, Pennsylvania 16802, USA, email: (ark5514, skd5685@psu.edu).} 
}      
\begin{document}
\maketitle






\begin{abstract}
Distributed Parameter Systems (DPSs), modelled by partial differential equations (PDEs), are increasingly vulnerable to disturbances arising from various sources. Although detection of disturbances in PDE systems have received considerable attention in existing literature, safety control of PDEs under disturbances remains significantly under-explored. In this context, we explore a practical input-to-state safety (pISSf) based control design approach for a class of DPSs modelled by linear Parabolic PDEs. Specifically, we develop a control design framework for this class of system with both safety and stability guarantees based on control Lyapunov functional and control barrier functional. To illustrate our methodology, we apply our strategy to design a thermal control system for battery modules under disturbance. Several simulation studies are done to show the efficacy of our method.
\end{abstract}

\section{Introduction}


Stability and safety verification in Ordinary Differential Equations (ODEs) has been explored widely and a detailed list of works can be found in survey papers \cite{sontag2008input} and \cite{ODEsurvey}. Two different notions have emerged in existing literature that characterize stability and safety, respectively: 
\begin{itemize}
    \item \textit{Input-to-state stability (ISSt) \cite{sontag1996new}:} Here is the objective is to study the stability of systems in the presence of external \textit{input}, and examine the boundedness of the system state trajectories only in a sense proportional to the size of the \textit{input}.
    \item \textit{Input-to-state safety (ISSf) \cite{jaya_iss2016new,kolathaya2018input,krstic2021inverse}:} Here the objective is to ensure that the system state trajectories stay away from a predefined unsafe region, or in other words, stay \textit{close} to safe region. Specifically, trajectories moving from safe zone towards unsafe region will violate safety boundary only in a sense proportional to the size of \textit{input}. On the other hand, trajectories starting in the unsafe region will be brought close of the safety boundary where \textit{closeness} is proportional to the size of the \textit{input}.
\end{itemize}
In the context of this work, we consider disturbances as external \textit{inputs} in the notions of ISSt and ISSf. ISSt in sense of Sontag has been investigated using Lyapunov functionals \cite{sontag1996new}. The notion of pratical ISSt has been explored in \cite{sontag1996new} that augments the original ISSt with certain practical considerations. On the other hand, ISSf was introduced in \cite{WIELAND2007462}. Since then, two prominent methods are generally used for ISSf analysis/design/verification: reachable sets approximation  \cite{limon2009input,angeli2015characterizations,coutinho2011nonlinear} and barrier functionals \cite{kolathaya2018input,xu2015robustness,wills2004barrier,taylor2020learning}. Among the various works in this domain, \cite{jaya_iss2016new,romdlony2016stabilization} introduced the notion of ISSf with respect to systems while \cite{kolathaya2018input} introduced the ISSf notion with respect to sets. Based on ISSf, inverse optimal safety filters have been explored in \cite{krstic2021inverse} which studies inverse optimality of safety filters under stochastic and deterministic disturbances.

Control of PDE systems has been widely explored over the years  \cite{krstic2008boundary,smyshlyaev2010adaptive,CHRISTOFIDES201321,karafyllis2018feedback}. Similar to ODEs, notions of ISSt for PDE systems have garnered a lot of attention recently (see survey paper \cite{mironchenko2020input}). For example, PDE ISSt have been explored for reaction-diffusion systems \cite{dashkovskiy2010uniform}, hyperbolic systems \cite{tanwani2017disturbance}, \cite{roy2020secure}, parabolic systems \cite{roy2021security}, parabolic PDE systems with boundary disturbances \cite{karafyllis2016iss}, \cite{karafyllis2019}, systems with distributed time-delays \cite{pepe2008liapunov}, and diffusion equation with time-varying distributed coefficients \cite{argomedo2012strict}. Notions of practical ISSt for PDEs have been explored in \cite{mironchenko2018criteria}. In contrast to ISSt, ISSf has remained mostly unexplored in the context of PDEs. In \cite{ahmadi2017safety}, safety verification using barrier functionals for homogeneous distributed parameter systems has been considered. In this work, numerical strategies based on semi-definite programming has been used for the construction of barrier functionals. {However, control performance under disturbances has not been considered in this work}. Given the importance of maintaining system safety under disturbances, it is critical to consider control system design for PDE systems under these disturbances. In \cite{koga2021safe}, safe control of Stefan system under disturbances is considered. In the framework proposed in \cite{koga2021safe}, an operator is allowed to manipulate the control input as long as safety constraints are satisfied; however, the safety control overrides the operator control signal realizing a feedback control ultimately guaranteeing safety. The feedback law for safety control is designed utilizing backstepping, quadratic programming, and a control barrier function. In our current work, we attempt an alternate approach to achieve safety control of a class of linear parabolic PDEs under disturbances. Specifically, we design a control law that employs feedback from the boundaries and an in-domain point, by utilizing a practical ISSf (pISSf) barrier functional characterization (inspired by the notion presented in \cite{jaya_iss2016new}). Subsequently, utilizing ISSt Lyapunov functional characterization, we prove that such designed safety control is also an input-to-state stabilizing control under certain additional conditions. In this way, we ultimately propose a feedback control law that satisfies the conditions of both ISSt and pISSf.

In light of the aforementioned discussion, the main contributions of this paper is the following: \textit{Building upon the existing literature, we extend PDE safety research by designing a feedback based control that satisfies both pISSf and ISSt under disturbances, utilizing pISSf barrier functional characterization and ISSt Lyapunov characterization.} As a case study, we consider a one-dimensional thermal PDE model for a battery module with a boundary coolant control. Next, we construct a control barrier functional and control Lyapunov functional for obtaining analytical guarantees for safety and stability for the battery system. The analytical guarantees allows us design the controller gains for actuating the boundary coolant. The rest of the paper is organized as follows. Section 2 sets up the problem by discussing the battery module thermal model and formulating control objectives. Sections 3 and 4 detail the pISSf-ISSt framework. Section 4 presents case studies to illustrate the proposed framework. Finally, Section 5 concludes the paper.




\section{Problem setup}



In this work, we consider a linear Parabolic PDE of the following form:
\begin{align}
    &  T_{t}(x,t)  =  \alpha T_{xx}(x,t) + Q(x,t) + \Delta(x,t), \label{eq1}
\end{align}
where $\alpha>0$, {$T:[0,L]\times[0, t_{max}]\to \mathbb{R}$ denotes the solution of the PDE over time $t\in [0,t_{max}], 0<t_{max}<\infty$ and space $x\in[0,L]$. Here $T(x,t)$ belongs to the space of twice differentiable functions over domain $D=[0,L]\times[0, t_{max}]$. Moreover, the known distributed input, denoted by $Q(x,t)$ and a disturbance, denoted by $ \Delta(x,t)$, belong to the space of continuous functions over domain D.} In this work, we consider Robin-type boundary condition (which can be regarded as a weighted sum of Neumann and Dirichlet conditions) given by
\begin{align}
    & T_{x}(0,t)=k\big(T(0,t)-T_{c1}(t)\big), \label{eq3}\\
    & T_{x}(L,t)= k\big(T_{c2}(t)-T(L,t)\big),\label{eq4}
\end{align}
where $k\in\mathbb{R}^+$ is a known parameter and  $T_{c1}(t)$ and $T_{c2}(t)$ are the boundary control inputs. {Finally, the initial condition for the system is given by $T(x,0) = T_0\in \mathbb R, \forall x \in [0,L]$.} 


We consider a battery module as a case study in this work. In the context of battery module, the one-dimensional PDE given by \eqref{eq1} captures the spatio-temporal thermal dynamics (please refer to \cite{chung2019thermal} and \cite{9483248} for details of the modeling). Here $T(x,t) $ is the distributed temperature of the module over the domain $D=[0,L]\times [0,t_{max}]$. {Next, let us denote the space of continuous functions over domain $X$ to be $C(X)$ and the set of positive real numbers to be $\mathbb{R}^+$. Then,  $\alpha=({k_{b}})/({\rho_{b} c_{p,b}})\in \mathbb{R}^+$ is the thermal conductivity, $Q = ({I^2(x,t) R})/({\rho_{b} c_{p,b}V_{b}})\in C(D)$ is the internal distributed heat generation in the battery module due to nominal current flow with $R$ being the resistance of the battery and $I(x,t)\in C(D)$ being the battery module current, and {$V_{b}\in \mathbb{R}^+$ represents volume of the module}. The variable $\Delta(x,t)\in C(D)$ represent an unknown disturbance arising either from a non-malicious physical phenomenon or by a malicious adversarial cyber-attack. Additionally, $\rho_{b}, c_{p,b},k_b\in \mathbb{R}^+$ represent the equivalent density, specific heat capacity and thermal conductivity of the battery, respectively. In terms of output measurements, temperature sensors are placed along the length of the battery module. In this work, we assume that we have sensors are placed in each boundary, and an additional sensor is placed in the middle of the module. That is, $T(0,t)$, $T(L,t)$, and $T(m,t)$ are measured where $m$ is a middle point, i.e. $0<m<L$. We note here that $T(0,t)$, $T(L,t),T(m,t)\in C([0, t_{max}])$.}

\subsection{Control Problem Formulation}

Effective thermal management requires maintaining the temperature along the battery module close to a desired safe temperature, that is $\left|T(x,t)-T_{d}\right|$ should be as small as possible, where $T_d$ is a predetermined desired safe temperature. In this work, we further characterize this objective in terms of pISSf and ISSt properties. First, we write the error system PDE as 
\begin{align}\label{h_system}
    &h_t(x,t) = \alpha h_{xx}(x,t) + D(x,t),\\\label{bc1}
    &h_x(0,t) = k[h(0,t)-u_1(t)],\\\label{bc2}
    &h_x(L,t) = k [u_2(t)-h(L,t)],
\end{align}
where $h(x,t)=T(x,t)-T_{d}$ is the error with respect to the desired safe temperature, and $u_1$ and $u_2$ are boundary cooling inputs. Here $D=\bar{D}+Q$ where $\bar{D}$ represents the disturbance (effect of $\Delta(x,t)$ in \eqref{eq1}) and $Q$ is the previously defined nominal input. We have combined both of them in a single notation to simplify the further analysis.  {In this work, we choose the following structures of the boundary cooling control inputs:
\begin{align}\label{input1}
    &u_1(t) = \mu h(0,t) ,\quad u_2(t) = \beta h(L,t),
\end{align}
where $\mu$ and $\beta$ are the control gains.} This modifies the boundary condition to 
\begin{align}\label{bc11}
    & h_x(0) = k (1-\mu)h(0), \quad h_x(L) = k(\beta-1)h(L).
\end{align}
Under this setting, our goal is to design these gains $\mu$ and $\beta$ such that the system is stabilized and safety is maintained. We define the following properties which provide us more precise notions of stability and safety.

\textbf{Safety Property} is given by \textit{pISSf} with respect to an unsafe set ${\mathscr{U}}$, which is defined by the following criterion on the system state $h(x,t)$ \cite{jaya_iss2016new}:
\begin{align}\label{issf}
   {\left|h(.,t)\right|_{\mathscr{U}}^2 \geqslant \bar{k}_1e^{\bar{k}_2t}|h_0|^2_{\mathscr{U}}
    -\bar{k}_3e^{\bar{k}_4t}\|D\|_\infty^2-\bar{k}_5e^{\bar{k}_6t},}
\end{align}
for $t\in [0,t_{max}]$, where $|.|_{\mathscr{U}}$ denotes a distance metric from the unsafe set; $\|.\|_\infty$ denotes the $L^\infty$  norm given by $\|M\|_\infty = (ess) \sup_{t\in [0,t_{max}]}\|M(x,t)\|^2_2$ with $0< t_{max}<\infty$; $\|.\|_2$ denotes the $L^2$ spatial norm given by $\sqrt{\int^L_0M^2(x,t)\diff{x}}$; $h_0 = h(.,0)$ is the initial condition for \eqref{h_system}; and $\bar{k}_i\in\mathbb{R}^+, i\in \{1, \cdots, 6\}$ are positive constants. This definition essentially states that the distance of the system states from the unsafe region is always lower bounded by the difference between two competing terms: one arising from initial distance while the other is dictated by the size of the disturbance \cite{jaya_iss2016new}. Keeping this distance to be positive will ensure that the system states always stay away from the unsafe set. {This means that the following should be ensured: $\bar{k}_2>\max\{\bar{k}_4,\bar{k}_6\}$.}


\textbf{Stability Property} is given by \textit{ISSt} which in turn is defined by the following criterion on the system state $h(x,t)$ \cite{jaya_iss2016new,mironchenko2018criteria}:
\begin{align}\label{isst}
   {\left\|h(.,t)\right\|_S^2  \leqslant \Tilde{k}_1e^{-\Tilde{k}_2t}\|h_0\|_S^2 +\Tilde{k}_3\|D\|_\infty^2,}
\end{align}
where $\left\|h\right\|_S$ is a spatial norm in $[0,L]$, defined here as  $\left\|h\right\|_S:=  \sqrt{\int_0^L h^2\diff{x} + h^2(L)+h^2(0)}$; $\|.\|_\infty$ denotes the $L^\infty$ norm defined after \eqref{issf}; $h_0 = h(.,0)$ is the initial condition for \eqref{h_system}; and $\Tilde{k}_i\in\mathbb{R}^+, i\in \{1, \cdots, 3\}$ are positive constants. The definition essentially states that system states (in the sense of certain norm) will be upper bounded by a combination of two terms: one arising from the initial states while the other arising from the disturbance. Under asymptotic conditions, the boundedness of the states will be dictated by the size of disturbance.  

\begin{remm}\normalfont
Note that the finite time horizon $t\in [0,t_{max}]$ imposed in this definition is inspired by the \textit{limited duration safety} explored in \cite{ohnishi2021constraint}. As mentioned in \cite{ohnishi2021constraint}, satisfying safety constraints over infinite time may lead to restrictions in design, which can be relaxed by formulating a \textit{limited duration safety} problem.
\end{remm}

\begin{remm}\normalfont
The parameter $\bar{k}_5$ essentially makes the safety definition \eqref{issf} a practical ISSf. That is, with $\bar{k}_5=0$, \eqref{issf} becomes ISSf condition. From the design point of view, as argued in \cite{jaya_iss2016new}, $\bar{k}_5$ can help accommodate a polynomial type barrier certificate via sum-of-squares design. On the other hand, from a robustness to model uncertainty point of view, $\bar{k}_5$ can also help accommodate model uncertainties -- in addition to the external disturbance. Nevertheless, we will design our controller such that $\bar{k}_5$ is made arbitrarily small.
\end{remm}

In the subsequent sections, our approach of finding the control gains are as follows. First, in Section 3, we find the conditions on control gains that satisfy the pISSf criterion in \eqref{issf}. Next, in Section 4, we show that the pISSf conditions on control gains additionally guarantee ISSt for the system in the sense of \eqref{isst}. 

\section{Input-to-State safety based control design}

In this section, we focus on the Input-to-State Safety condition mentioned in \eqref{issf}. First, we formulate the unsafe set $\mathscr{U}$ and the distance metric $\left|h\right|_{\mathscr{U}}$. Subsequently, we construct a control barrier functional that would eventually help us derive the conditions on control gains to satisfy \eqref{issf}.

\subsection{Unsafe set and distance metric formulation}
The goal of this work is to design a control strategy that will guarantee safety of the battery system under anomalies. The criterion for safety is that the spatial norm of the temperature deviation of the battery from a set-point remains below a prescribed threshold $\overline{h}$. Mathematically, this implies 
\begin{align}\label{main_safety}
    \left\{\int_0^Lh^2(x,t)\diff{x}\right\}^\frac{1}{2}\leqslant \overline{h}, \forall x\in [0,L], t\in [0,t_{max}].
\end{align}
Alternatively, we can define  an unsafe set
\begin{align}\label{unsafe_set}
    \mathscr{U} = \{a\in \mathbb{R}: a> \overline{h}\},
\end{align}
such that $\left\{\int_0^Lh^2(x,t)\diff{x}\right\}^\frac{1}{2}> a,  \forall a \in  \mathscr{U}, t\in [0,t_{max}]$ .

Keeping this unsafe set in mind, let us define the distance metric for our framework.

For a given function $h\in C([0,L]\times [0,t_{max}])$ with $0<t_{max}<\infty$, we define the distance of $h(x,t)$ from an unsafe  set $\mathscr{U}$ (defined in \eqref{unsafe_set}) by the following metric:
 \begin{align}\label{dist_metric}
      \hu:=\inf\limits_{a\in {\mathscr{U}}}&\left\{a^2 -\int_0^Lh^2(x,.)\diff{x}\right\}^\frac{1}{2}.
 \end{align}

\subsection{pISSf via barrier functional characterization}
It can be shown that the existence of a particular pISSf barrier functional can automatically guarantee pISSf in the sense of \eqref{issf} \cite{jaya_iss2016new}. For the sake of completeness, we present the following proposition which is similar to the one presented in \cite{jaya_iss2016new}.

\vspace{0.1in}
\begin{propp}\label{existence}
Consider the PDE system given by \eqref{h_system}-\eqref{input1}, the prescribed unsafe set $\mathscr{U}\subset \mathbb{R}^+$ given by \eqref{unsafe_set} and the distance metric as defined in \eqref{dist_metric}. Suppose there exists a safety barrier functional $\mathscr{B}: \mathbb{H}_1 \to \mathbb{R}$ satisfying the following two conditions:\\
\noindent
  {\textbf{Condition 1:}  $-c_1\hu^2 -\rho \leqslant \bh \leqslant -c_2 |h|^2_{\mathscr{U}}$, and\\
\noindent
\textbf{Condition 2:} 
        $\bdot \leqslant -c_3|h|^2_{\mathscr{U}} +c_4 \|D\|^2$,}\\
{where $ \mathbb{H}_1$ represents the Sobolev space containing square-integrable functions whose derivatives are also square-integrable.} Here $\bdot $ is the derivative of $\mathscr{B}$ along the solution trajectory $h(x,t)$ of the PDE system, and  $c_i, \forall i\in \{1,2,3,4,5\}$ and $\rho$ are positive constants. Then the PDE system \eqref{h_system}-\eqref{input1} is considered to be practically input-to-state safe (pISSf) with respect to the unsafe set $\mathscr{U}$.  
\end{propp}
\vspace{-0.1in}
\begin{proof}
This proof can be done by following the approach shown in \cite{jaya_iss2016new}. Note that the parameter $\rho$ captures the effect of $\bar{k}_5$ presented in pISSf definition \eqref{issf}. 
\end{proof}


According to the above proposition, if we can construct a safety barrier functional $\mathscr{B}$, we can guarantee the pISSf property for the PDE system given in \eqref{h_system}. Let us now construct a pISSf barrier functional $\mathscr{B}$ that satisfies the two conditions presented in Proposition \ref{existence}. 

\subsection{Construction of pISSf barrier functional and design requirements}


Now, before presenting the theorem that provides the design specifications for ensuring pISSf, we will introduce a version of Poincare's Inequality through the following Lemma.
\vspace{0.1in}
\begin{Lemm}[A Version of Poincare Inequality]\label{Poincare_mod0}
For the PDE system given by \eqref{h_system} with continuously differentiable solution on $x\in[0,L]$, the following inequaltity is true:
\begin{align}\label{Poincare_mod1}
        &-\int_0^L\!\!\!h_x^2\diff{x} \leqslant  \frac{1}{4L}\left[h^2(0)+h^2(L)\right] -\frac{1}{4L^2}\int_0^L\!\!\!h^2\diff{x}.
\end{align}

\end{Lemm}

\begin{proof}

First, let us use integration by parts on the following integral to obtain:
\begin{align}
    \int_0^L h^2\diff{x} = Lh^2(L)-2\int_0^L xhh_x\diff{x}.\label{d1}
\end{align}
Now, applying Young's inequality on the second term of right hand side of \eqref{d1} yields:
\begin{align}
     \int_0^L h^2\diff{x} \leqslant Lh^2(L) + \frac{1}{\sigma_2}\int_0^L h^2\diff{x}+{\sigma_2}\int_0^Lx^2 h_x^2\diff{x}. \label{d2}
\end{align}
where $\sigma_2$ is an arbitrary positive constant. Since $x\in [0,L]$, we can majorize with $x^2 \leqslant L^2$ and  re-arrange terms in \eqref{d2} to get:
\begin{align}
    \left(1-\frac{1}{\sigma_2}\right)\int_0^Lh^2\diff{x}\leqslant  Lh^2(L) +{\sigma_2}L^2\int_0^L h_x^2\diff{x}. \label{d3}
\end{align}
Re-arranging terms of \eqref{d3} furthermore yields the following: 
\begin{align}\label{lemma1a}
      &-\int_0^Lh_x^2\diff{x} \leqslant  \frac{1}{L}\frac{1}{\sigma_2}h^2(L) +\frac{1-\sigma_2}{\sigma_2^2L^2}\int_0^Lh^2\diff{x}.
\end{align}

Next, by integrating by parts the following integral, we get
\begin{align}\nonumber
    -\int_0^L&h^2(L-p)\diff{p}= Lh^2(0)\\&+2\int_0^L xh(L-p)h_x(L-p)\diff{p}.\label{d5}
\end{align}
Now we apply Young's inequality to the second term of the right hand side of \eqref{d5}, and further using  $p^2\leqslant L^2$, we obtain:
\begin{align}\nonumber
    -\int_0^Lh^2(L-p)\diff{p}\leqslant &Lh^2(0)+\frac{1}{\sigma_1}\int_0^L h^2(L-p)\diff{p}\\&+{\sigma_1}L^2\int_0^L h^2_x(L-p)\diff{p},
\end{align}
where $\sigma_1$ is an arbitrary positive constant. 
Now, defining $L-p:=x$, we can re-write the above inequality as
\begin{align}\nonumber
    \int_0^Lh^2(x)\diff{x}\leqslant& Lh^2(0)+\frac{1}{\sigma_1}\int_0^L h^2(x)\diff{x}\\&+{\sigma_1}L^2\int_0^L h^2_x(x)\diff{x}.\label{d7}
\end{align}
Again, re-arranging terms of \eqref{d7}, we can obtain:
\begin{align}
       &-\int_0^Lh_x^2\diff{x} \leqslant  \frac{1}{L}\frac{1}{\sigma_1}h^2(0) +\frac{1-\sigma_1}{\sigma_1^2L^2}\int_0^Lh^2\diff{x}.
  \label{lemma1b}
\end{align}
Subsequently, from \eqref{lemma1a} and \eqref{lemma1b}, we obtain:
\begin{align}\nonumber
    -\int_0^Lh_x^2\diff{x} & = -\frac{1}{2}\int_0^Lh_x^2\diff{x} -\frac{1}{2}\int_0^Lh_x^2\diff{x}\\\nonumber
    & \leqslant 
    \frac{1}{2L}\left[\frac{1}{\sigma_1}h^2(0)+\frac{1}{\sigma_2}h^2(L)\right]\\\label{lemma1c} &\hspace{3mm}+\left[\frac{1-\sigma_1}{2\sigma_1^2L^2}+\frac{1-\sigma_2}{2\sigma_2^2L^2}\right]\int_0^Lh^2\diff{x}.
\end{align}
Thus, plugging in $\sigma_1=\sigma_2=2$ in \eqref{lemma1c}, we obtain \eqref{Poincare_mod1}. 
This completes the proof of this lemma.
\end{proof}
\vspace{-2mm}

 Next, in the following Theorem, we show that the control barrier function given by \eqref{barrier_ex}  guarantee pISSf in the sense of \eqref{issf} by satisfying the two conditions presented in \mbox{Proposition 1}.

\vspace{0.1in}
\begin{thmm}[Design Requirements for pISSf]
 {Consider the system \eqref{h_system} with boundary conditions \eqref{bc11}. Let us also consider the unsafe set for this system to be \eqref{unsafe_set} and the metric measuring the distance from this unsafe set to be given by \eqref{dist_metric}. If the controller gains are chosen such that the following inequalities are satisfied,
\begin{align}\label{gains1}
   \left[(\mu-1)+\frac{1}{4Lk}\right]\leqslant0,\quad
     \left[(\beta-1)+\frac{1}{4Lk}\right]\leqslant0,
\end{align}
then the system \eqref{h_system} satisfies the two conditions of \mbox{Proposition 1}, and is considered to be practical Input-to-State Safe (pISSf) with respect to the unsafe set $\mathscr{U}$. }
\end{thmm}

\begin{proof}
First, let us this construct the pISSf barrier functional:
\begin{align} 
    \bh= & \int_0^L\!\!h^2(x,.)\diff{x}-\overline{h}^2-\kappa_1\overline{h}^2e^{-\kappa_2(t_{max}-t)},  \label{barrier_ex}
\end{align}
$\forall t\in[0, t_{max}]$. Here, $\kappa_1>0$ and $\kappa_2=1/\kappa_2$ are positive constants.
Now, we will prove that the barrier functional \eqref{barrier_ex} satisfies Condition 1 from Proposition \ref{existence}.  For the right-hand side of the inequality in Condition 1, we majorize the functional given in \eqref{barrier_ex} by neglecting the exponential term to obtain:
\begin{align}
    \bh\leqslant\left(\int_0^L\! h^2\diff{x}-\overline{h}^2\right).
\end{align}
We note here that $\overline{h}^2\geqslant \inf_{a\in \mathscr{U}}a^2$ from \eqref{unsafe_set}. Next, considering the distance metric in \eqref{dist_metric}, we can write 
\begin{align}
    \bh \leqslant -\inf\limits_{a\in \mathscr{U}}\left\{a^2-\int_0^L\!\!\!h^2\diff{x} \right\}\leqslant-c_2 \hu^2.
\end{align}
where $c_2\leqslant 1$. 

Thereafter, we prove the left-hand side of the inequality in Condition 1 from Proposition \ref{existence}. We observe that $ \overline{h}^2e^{-\kappa_2t_{max}}\leqslant \overline{h}^2e^{-\kappa_2(t_{max}-t)}\leqslant \overline{h}^2, \forall t\in [0, t_{max}]$. 
Here, using the upper bound on the exponential, we can similarly write
\begin{align}\label{lastb}
    \bh \geqslant \int_0^L h^2\diff{x}-\overline{h}^2-\kappa_1\overline{h}^2.
\end{align}
Then, we use \eqref{unsafe_set} to conclude that the distance from the unsafe set boundary is always greater than the distance from the inside of the set. Mathematically, this implies, $\overline{h}^2-\int_0^L h^2\diff{x}\leqslant \inf\limits_{a\in \mathscr{U}}\left\{a^2-\int_0^Lh^2\diff{x} \right\}=\hu$. Thus, \eqref{lastb} yields
\begin{align}
    \bh \geqslant  -c_1 \hu^2 -\rho.
\end{align}
where $c_1\geqslant 1$, and $\rho=\overline{h}^2\kappa_1$.  {Note that $\rho$ can be made arbitrarily small by choosing an arbitrarily small $\kappa_1$.} 
This confirms that $\bh$ satisfies Condition 1 from Proposition \ref{existence}. 


Next, we will prove that  the barrier functional \eqref{barrier_ex} satisfies Condition 2 from Proposition \ref{existence}. In order to do so, we take the derivative of \eqref{barrier_ex} along $h(x,t)$ to obtain:
\begin{align}
    \bdot =&2\int_0^L hh_t\diff{x}-\kappa_1\kappa_2\overline{h}^2e^{-\kappa_2(t_{max}-t)}.
\end{align}
Here, we note that $ \overline{h}^2e^{-\kappa_2t_{max}}\leqslant \overline{h}^2e^{-\kappa_2(t_{max}-t)}, \forall t\in [0, t_{max}]$ and $\kappa_1\kappa_2=1$. Then, using this lower bound on the exponential once again, we can obtain the following inequality
\begin{align}
    \bdot \leqslant 2\int_0^L hh_t\diff{x}-\kappa_1\kappa_2\mathcal{E}\overline{h}^2. \label{d111}
\end{align}
where $\mathcal{E}=e^{-\kappa_2t_{max}}$. Subsequently, we replace $h_t$ in \eqref{d111} using \eqref{h_system}, and further using integration by parts we have
\begin{align}\label{bdot}
    \bdot \leqslant 2 BT_1 -2\alpha \int_0^L\!\!\!\!\! h_x^2  +2\int_0^L h D -\mathcal{E}\overline{h}^2,
\end{align}
where 
\begin{align}\label{ssare11}
    BT_1 = &\alpha\left[ h_x(L)h(L)-h_x(0)h(0)\right]
\end{align}
By replacing $h_x(0)$ and $h_x(L)$ from \eqref{bc11}, in \eqref{ssare11}, we get
\begin{align}\label{BT1}
    BT_1 = &k\alpha \left[(\beta-1) h^2(L) +(\mu-1) h^2(0)\right]
\end{align}
Next, we apply Young's inquality on the third integral of the right hand side of \eqref{bdot} to obtain,
\begin{align}\label{bdot111}
    \bdot \leqslant& 2BT_1 -2\alpha \int_0^L\!\!\!\!\! h_x^2 +\nu_1\int_0^L\!\!\!\!\! h^2  + \frac{1}{\nu_1} \int_0^L\!\!\!\!\! D^2-\mathcal{E}\overline{h}^2,
\end{align}
where we choose $\nu_1=\frac{\alpha}{2L^2}>0.$ Next, we multiply \eqref{Poincare_mod1} with $2\alpha$ in Lemma 1 to obtain: 
 \begin{align}\label{dxxx}
 -{2\alpha}\int_0^L\!\!\!\!\!h_x^2\diff{x} \leqslant  \frac{\alpha}{2L}\left[h^2(0)+h^2(L)\right] -\frac{\alpha}{2L^2} \int_0^L\!\!\!\!\!h^2\diff{x}. 
 \end{align}
 Thus, applying \eqref{dxxx} in \eqref{bdot111}, cancelling $\int_0^L h^2$ term from \eqref{bdot111} and substituting the value of $\nu_1$ yields:
\begin{align}\label{bdot11}
    \bdot \leqslant & BT_2+c_4 \int_0^L\!\!\!\!\! D^2-\mathcal{E}\overline{h}^2,
\end{align}
where $c_4=\frac{2L^2}{\alpha} $ and  $BT_2$ is written as
\begin{align}\nonumber
BT_2=&2\alpha\Bigg[\left(k(\beta-1)+\frac{1}{4L}\right)h^2(L)\\
&+\left(k(\mu-1)+\frac{1}{4L}\right)h^2(0)\Bigg]
\end{align}
%
%
Now, if we choose the gains $\mu$ and $\beta$ using the constraints in \eqref{gains1}, then we can majorize $\bdot$ by neglecting it. Furthermore, we can add a term $\mathcal{E}\int_0^Lh^2$ to $\bdot$ to further majorize it,  and obtain:
\begin{align}\label{bdotxx}
    \bdot \leqslant \mathcal{E}\int_0^L\!\!\!\!\! h^2 -\mathcal{E}\overline{h}^2+ c_4 \int_0^L\!\!\!\!\! D^2.
\end{align}
Thus, using the definition of distance metric and choosing $c_3=\mathcal{E}>0$ yields
 \begin{align}
    \bdot \leqslant -c_3 \hu^2+c_4 \left\| D\right\|^2.
\end{align}
 This completes our proof.
\end{proof}

In this section, we have derived the conditions on control gains for which the system is pISSf. In the following section, we will show that the derived conditions for pISSf ensures ISSt for the system.
\section{Input-to-State stabilizing safe control}

In the present section, we will show that the control gain conditions in Theorem 1 simultaneously satisfy the pISSf criterion in the sense of \eqref{issf} and  ISSt criterion in the sense of \eqref{isst}. 
Following the results in existing literature \cite{mironchenko2017characterizations}, we can say that if there exists a functional $V(h)$ for an infinite dimensional system, which satisfies the following two conditions:
 \begin{align}\label{vcond1}
     &d_1\|h\|_S^2\leqslant {V}(h) \leqslant d_2 \|h\|_S^2,\\\label{vcond2}
    &\dot{{V}}(h) \leqslant -d_3V(h) +d_4 \|D\|^2
 \end{align}
 where $d_i, \forall i\in \{1,2,3,4,5\}$ are positive constants, then system is ISSt. In our formulation, the norm is defined as $\left\|h(.,t)\right\|_S:= \sqrt{\int_0^L h^2(x,t)\diff{x}}$.
 


Now, we will present the theorem which prescribes the design requirements on the controller gains in order to guarantee both pISSf and ISSt for the PDE system \eqref{h_system}-\eqref{input1}.

\vspace{0.1in}
\begin{thmm}[Design Requirements for both pISSf and ISSt]
 {Consider the system \eqref{h_system} with boundary conditions \eqref{bc11}. If there exists controller gains that satisfy pISSf inequality conditions given in \eqref{gains1}, 
then the system \eqref{h_system} is considered to be both pISSf and ISSt.}
\end{thmm}

\begin{proof}
Considering the following Lyapunov functional:
\begin{align}\label{Lyapunov}
    V(h) =\frac{1}{2}\int_0^L h^2\diff{x}. 
\end{align}
The first condition from \eqref{vcond1} can be proved by choosing $0<d_1 <\frac{1}{2},$ and $d_2 >\frac{1}{2}$.

Next, our goal is to prove the second condition given in \eqref{vcond2}. We take the time derivative of $V(h)$ along the direction of the solution of \eqref{h_system} to get
\begin{align}\label{Lyapunovxde}
   \dot{V} =\int_0^L h h_t\diff{x}=\int_0^L (\alpha hh_{xx}+hD)\diff{x}. 
\end{align}
Subsequently, we apply integration by parts to the first  term and Young's inequality to the second term of the right hand side of \eqref{Lyapunovxde}, respectively. Moreover, we note here that $-\alpha\int_0^L\!\!\!h_x^2\diff{x}<-\frac{\alpha}{4}\int_0^L\!\!\!h_x^2\diff{x}$ and re-write \eqref{Lyapunovxde} as 
\begin{align}\label{Vdot}
    \dot{V}=& BT_1 -\frac{\alpha}{4}\int_0^L\!\!\!h_x^2\diff{x}+ \frac{\alpha}{8L^2}\int_0^L\!\!\!h^2 \diff{x}+\frac{2L^2}{\alpha}\|D\|^2,
\end{align}
where $BT_1$ is given in \eqref{ssare11} which simplifies to \eqref{BT1} using \eqref{bc11}. 
Next, we multiply \eqref{Poincare_mod1} by $\frac{\alpha}{2}$ to obtain:
\begin{align}\label{PIarr}
 -\frac{\alpha}{4}\int_0^L\!\!\!\!\!h_x^2\diff{x} \leqslant  \frac{\alpha}{16L}\left[h^2(0)+h^2(L)\right] -\frac{\alpha}{16L^2} \int_0^L\!\!\!\!\!h^2\diff{x}.
\end{align}
Furthermore, we applied  to the first integral in \eqref{Vdot} to obtain:
\begin{align}\label{Vdott}
    \dot{V}=&BT_2-\frac{\alpha}{8L^2}\int_0^L\!\!\!h^2 \diff{x}+\frac{2L^2}{\alpha}\|D\|^2,
\end{align}
where 
\begin{align}\label{BT3}\nonumber
BT_2=&2\alpha\Bigg[\left(k(\beta-1)+\frac{1}{8L}\right)h^2(L)\\
&+\left(k(\mu-1)+\frac{1}{8L}\right)h^2(0)\Bigg].
\end{align}
%
Now, it is evident that $BT_2<0$,  if 
\begin{align}\label{pISSt}
    \left[(\mu-1)+\frac{1}{8Lk}\right]<0, \quad 
    \left[(\beta-1)+\frac{1}{8Lk}\right]<0.
\end{align}
We also note here that if the gains  satisfy the pISSf conditions given in \eqref{gains1}, then \eqref{pISSt} will be automatically satisfied and $BT_2<0$. 
Thus, we can majorize \eqref{Vdott} and using \eqref{Lyapunov} yields:
\begin{align}\label{Vdott1}
    \dot{V}\leqslant -{d}_3V+{d}_4 \|D\|^2,
\end{align}
where ${d}_3=\frac{\alpha}{4L^2}$ and ${d}_4 = \frac{2L^2}{\alpha}$. Thus, we proved that the pISSf condtions derived in \eqref{gains1}, also satisfy the less restrictive ISSt condition \eqref{pISSt}. This implies that if the gains satisfy \eqref{gains1}, the system will be considered both pISSf and ISSt. 

\end{proof}

\section{Simulation case studies}

In this section, we present some case studies to illustrate the effectiveness of the proposed framework. 
The battery model parameters are adopted from \cite{karimi2013thermal,kretzschmar2019tables}. The simulation is implemented in MATLAB. In the plant model simulation, we have injected zero mean Gaussian type measurement noise ($\mathcal{N}(0,0.1)$) in the temperature outputs as well as zero mean Gaussian type process noise ($\mathcal{N}(0,0.01)$) in the system dynamics to capture realistic scenarios. The battery module is operated under a current profile derived from Urban Dynamometer Driving Schedule (UDDS).
The desired set-point temperature here is $T_d=298K$.

To illustrate the advantages of the proposed approach, we compare the following approaches:
\begin{itemize}
    \item \textbf{Stability-Only Control (St-C):} In this case, we have used ISSt criterion \eqref{pISSt} to design the closed-loop control gains. Note that, if the design is done solely based on ISSt criterion, then there is no guarantee that it will also satisfy pISSf. To illustrate this point, and to highlight the potential advantage of combined pISSf-ISSt design, here we choose the St-C such that the gains do not satisfy pISSf conditions for the chosen unsafe set.
    \item \textbf{Stability-and-Safety Control (StSf-C) - Proposed Approach:} We have used pISSf and ISSt criteria given in \eqref{gains1} and \eqref{pISSt}, respectively, to design the closed-loop control gains. 
\end{itemize}

\begin{figure}[ht]
    \centering
    \includegraphics[width = 0.49\textwidth]{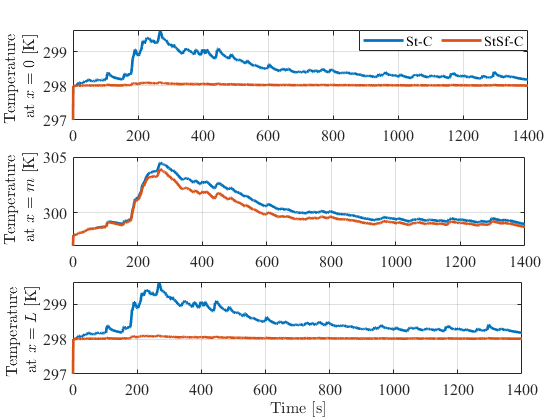}%
    \caption{Temperature measurement at the two boundaries and the mid-section from the battery module under no-anomaly scenarios with  Stability-Only Control (St-C), and Stability-and-Safety Control (StSf-C).}
    \label{fig:battery_temp_NF}
\end{figure}
\begin{figure}[ht]
    \centering
    \includegraphics[width = 0.49\textwidth]{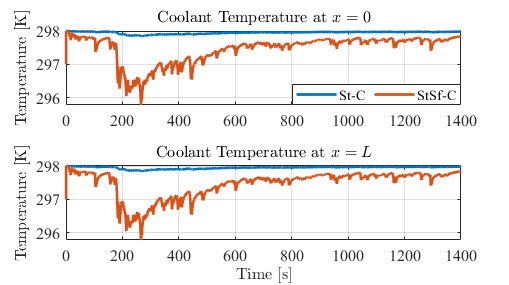}%
    \caption{Coolant temperatures at the two boundaries of the battery module under no-anomaly scenarios with Stability-Only Control (St-C), and Stability-and-Safety Control (StSf-C).}
            \label{fig:coolant_NF}
\end{figure}
The temperature response and control variables under no-anomaly scenario are shown in Figs. \ref{fig:battery_temp_NF} - \ref{fig:coolant_NF}.
The temperature of the battery module from the two boundaries and mid-section, as shown in Fig. \ref{fig:coolant_NF}, confirms the safety and stability of both strategies. In Fig. \ref{fig:coolant_NF}, the coolant temperature control at the boundaries show that the transient control action for StSf-C is greater than St-C. However, in steady state both control actions are somewhat comparable.

Next, we present a test case to illustrate the performance of the proposed approach under disturbance. We consider a scenario where an adversary injects a cyberattack in the form of a disturbance to the battery module to induce overdischarge. The disturbance is injected at $700s$ as \textit{current drain} from the module which forces the State-of-Charge (SOC) of the battery to reach zero. The SOC evolution under nominal scenario and under disturbance are shown in Fig. \ref{fig:SOC_attack}. It can be seen that the disturbance was initiated around $700s$, and consequently, after $1098s$ the modules goes into the overdischarge mode by crossing zero SOC. Furthermore, the overdischarge proceeds to induce a battery failure through increased heating of the cell \cite{overdischarge}. The increased heat generation due to additional current drain and subsequent heating due to overdischarge is shown in Fig. \ref{fig:heat_attack}.

\begin{figure}[ht!]
    \centering
    \includegraphics[width = 0.49\textwidth]{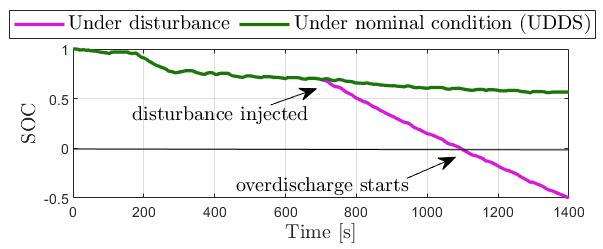}%
    \caption{State-of-Charge (SOC) of the battery module under nominal condition and under disturbance through battery overdischarge.}
            \label{fig:SOC_attack}
\end{figure}

\begin{figure}[ht!]
    \centering
    \includegraphics[width = 0.49\textwidth]{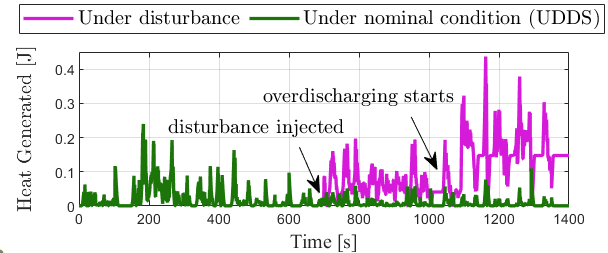}%
    \caption{Heat generated in the battery module under nominal condition and under disturbance through battery overdischarge.}
            \label{fig:heat_attack}
\end{figure}

Next, we present the temperature response and control variables under the two control strategies, as shown in \mbox{Figs. \ref{fig:PDE_cyb}-\ref{fig:coolant_cyb}}. The spatiotemporal temperature distribution of the battery module in \mbox{Fig. \ref{fig:PDE_cyb}} shows an increased in temperature at $700s$ for all two strategies after the disturbance injection by the adversary and a further rise in temperature after the overdischarge is initiated in the module. However, only with StSf-C scheme the temperature of the module remains under the unsafe zone while temperatures reach unsafe values for Sf-C. This is again corroborated by the temperature plots in \mbox{Fig. \ref{fig:battery_temp_cyb}}. The unsafe zone is again shown as the grey area (above $323K$) in the middle temperature plot  in \mbox{Fig. \ref{fig:battery_temp_cyb}}. Even though the measured temperatures at $x=0$ and $x=L$, remain in the safe zone for both strategies, the temperature at the midpoint of the battery clearly shows that  the St-C strategy have violated the safety condition while SfSt-C were able to maintain the battery module temperature under allowable maximum of $323K$. This is evident from the middle temperature plot in \mbox{Fig. \ref{fig:battery_temp_cyb}} where the temperature of the battery under St-C (shown in blue) crosses the unsafe zone.  Thus, this result shows the potential benefits of StSf-C.

\begin{figure}
    \centering
   {\includegraphics[width = 0.49\textwidth]{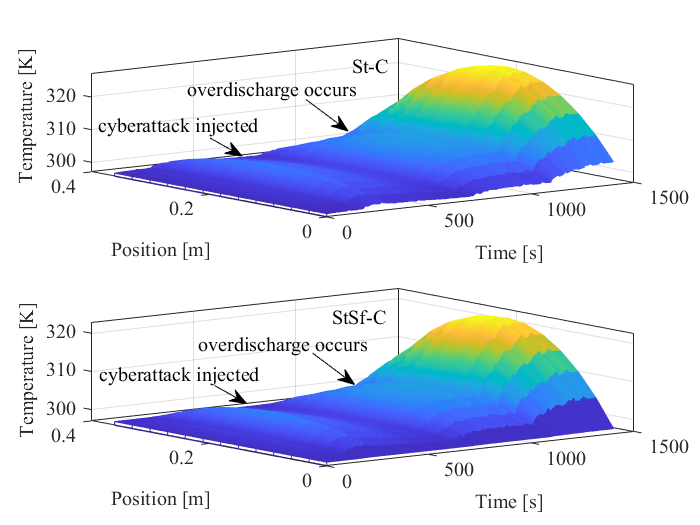}}%
\caption{Spatiotemporal temperature distribution in the battery module under overdischarge with Stability-Only Control (St-C), and Stability-and-Safety Control (StSf-C). }
    \label{fig:PDE_cyb}
\end{figure}
\begin{figure}
    \centering
    \includegraphics[width = 0.49\textwidth]{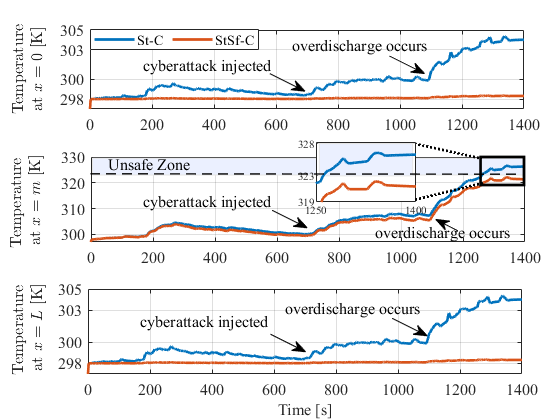}%
    \caption{Temperature measurement at the two boundaries and the mid-section from the battery module under overdischarge with Stability-Only Control (St-C), and Stability-and-Safety Control (StSf-C).}
    \label{fig:battery_temp_cyb}
\end{figure}


\begin{figure}[ht!]
    \centering
    \includegraphics[width = 0.49\textwidth]{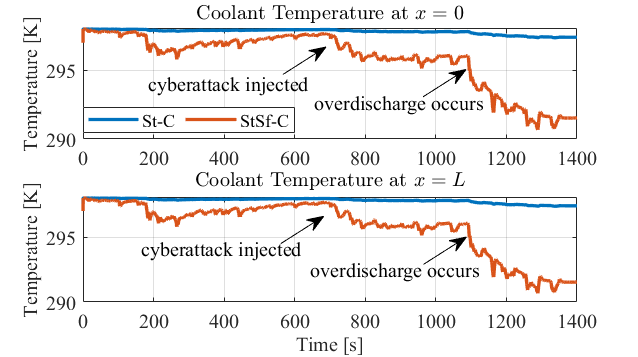}%
    \caption{Coolant temperatures at the two boundaries of the battery module under overdischarge in the module with  Stability-Only Control (St-C), and Stability-and-Safety Control (StSf-C).}
            \label{fig:coolant_cyb}
\end{figure}

\section{Conclusion}
In this paper, we have explored safe control of a class of linear Parabolic PDEs under disturbances. First, we defined unsafe sets and distance of the system states from such unsafe sets. Next, we constructed both control barrier and Lyapunov functional in order to develop a design framework for the controller under specific safety and stability guarantees. Additionally, we have applied our proposed strategy in the context of battery management system using boundary coolant control. We present the efficacy of our proposed methodologies through simulation studies under nominal conditions and disturbed conditions. The simulation study shows that the proposed approach can be beneficial to maintain safety limits.  {As a future work, we plan to extend the framework to (i) $n$-dimensional PDEs and apply it towards thermal management of large-scale battery packs, and (ii) PDEs with saturation on input magnitude and rates.}








\bibliography{ref1}

\end{document}